\documentclass[11pt,reqno]{amsart}
\usepackage{url,framed}
\usepackage[color=yellow]{todonotes}
\usepackage{setspace}
\usepackage{latexsym}
\usepackage{dsfont}

\usepackage{mathrsfs}
\usepackage[colorlinks]{hyperref}

\usepackage{amsfonts}
\usepackage{amsmath}
\usepackage{amscd}
\usepackage{amssymb}
\usepackage{geometry}
\usepackage{color}
\usepackage{graphicx}%
\def\p{{\partial }}
\def\bx{{\mathbf {x} }}
\def\by{{\mathbf {y} }}
\def\bu{{\mathbf {u} }}
\def\bv{{\mathbf {v} }}
\def\bA{{\mathbf {A} }}
\def\bB{{\mathbf {B} }}
\def\bD{{\mathbf {D} }}
\def\bE{{\mathbf {E} }}
\def\bH{{\mathbf {H} }}
\def\bJ{{\mathbf {J} }}
\def\bP{{\mathbf {P} }}
\def\bS{{\mathbf {S} }}
\def\bbH{{\mathbb {H} }}
\def\cH{{\mathcal {H} }}
\def\cL{{\mathcal {L} }}
\def\bsxi{{\boldsymbol {\xi} }}
\def\bsPi{{\boldsymbol {\Pi} }}
\def\bseta{{\boldsymbol {\eta} }}
\def\bsbeta{{\boldsymbol {\beta} }}
\def\bsgamma{{\boldsymbol {\gamma} }}
\def\bsomega{{\boldsymbol{\omega} }}
\def\bom{{\boldsymbol{\omega} }}
\def\bsvarpi{{\boldsymbol{\varpi} }}
\def\bsvt{{\boldsymbol{\tilde{v}} }}
\def\dd{{\color{red}\mathsf d}}
\def\de{\delta}
\providecommand{\U}[1]{\protect\rule{.1in}{.1in}}
\newtheorem{theorem}{Theorem}

\newtheorem{definition}[theorem]{Definition}

\newtheorem{remark}[theorem]{Remark}

\makeatletter
\@addtoreset{figure}{section}
\def\thefigure{\thesection.\@arabic\c@figure}
\def\fps@figure{h, t}
\@addtoreset{table}{bsection}
\def\thetable{\thesection.\@arabic\c@table}
\def\fps@table{h, t}
\@addtoreset{equation}{section}

\makeatother

\begin{document}

\title{Stochastic evolution of Augmented Born-Infeld Equations}
\author{Darryl D. Holm}

\address{DDH: Department of Mathematics, Imperial College, London SW7 2AZ, UK.}

\date{18 July 17}

\begin{abstract}
This paper compares the results of applying a recently developed method of stochastic uncertainty quantification designed for fluid dynamics to the Born-Infeld model of nonlinear  electromagnetism. The similarities in the results are striking. Namely, the introduction of Stratonovich cylindrical noise into each of their Hamiltonian formulations introduces stochastic Lie transport into their dynamics in the same form for both theories. Moreover, the resulting stochastic partial differential equations (SPDE) retain their unperturbed form, except for an additional term representing  induced Lie transport by the set of divergence-free vector fields associated with the spatial correlations of the cylindrical noise. 

The explanation for this remarkable similarity lies in the method of construction of the Hamiltonian for the Stratonovich stochastic contribution to the motion in both cases; which is done via pairing spatial correlation eigenvectors for cylindrical noise with the momentum map for the deterministic motion. This momentum map is responsible for the well-known analogy between hydrodynamics and electromagnetism. The momentum map for the Maxwell and Born-Infeld theories of electromagnetism treated here is the 1-form density known as the Poynting vector. Two Appendices treat the Hamiltonian structures underlying these results. 
\end{abstract}

\date{\normalsize \today
}

\maketitle

\tableofcontents

\doublespacing 

\section{Introduction}

Physics is an observational science. Hence, one may be led to consider how the modern stochastic methods for uncertainty quantification and data assimilation currently being developed for large scale observational sciences such as weather forecasting and climate change might be applied in foundational classical physics models, such as Euler's fluid vorticity equations and Maxwell's electromagnetic field equations. One might also wonder what mathematical differences may arise in the approaches for quantifying uncertainty in two such different foundational models, one concerning swirling fluids and the other concerning electromagnetic waves propagating in a vacuum. In addition, one might wonder about the role of mathematical structure in the formulation of stochastic methods of uncertainty quantification for two such different models. 

We will address these questions here by comparing the stochastic equations developed for quantifying uncertainty in the nonlinear dynamics of the ideal Euler fluid equations with the corresponding stochastic equations for the Born-Infeld electromagnetic field equations. Of course, the physics of these two models is fundamentally different. The deterministic Born-Infeld model arose in quantum field theory (QFT) and is closely related to string theory. It's QFT origins are discussed, e.g., in \cite{Bial1984,Bial1992} and its relationship to string theory is discussed in \cite{Pol1998,Gi2001}. Its analytical properties have been reviewed recently in \cite{Kiessling2011}. Its parallels with hydrodynamics are discussed in \cite{Arek1989,Br2002,Br2004,BrYo2005}. 

The theories being discussed here all share the same conceptual framework. Born-Infeld electromagnetism, string theory and ideal fluid dynamics are all Hamiltonian theories, whose symmetries enable reduction to variables that are invariant under a Lie group. In addition, the Born-Infeld field equations imply augmented equations for energy and momentum conservation that are reminiscent of conservation laws in fluid dynamics, as shown, e.g., in \cite{Arek1989,Br2002}. 

The shared Hamiltonian structure of the \emph{augmented Born-Infeld} (ABI) equations and ideal fluid dynamics will put the Hamiltonian approach discussed here for introducing stochastic uncertainty quantification methods into a common framework. Here, we will concentrate on introducing stochasticity by using the parallels between the Born-Infeld field equations and hydrodynamics. As a result, the stochastic version will turn out to conserve the deterministic Born-Infeld energy. The preservation of other analytical properties of the Born-Infeld field equations under the addition of this type of stochasticity will  be explored elsewhere. The possibilities for applying this structure-preserving stochastic Hamiltonian approach to estimating the unknown effects of unobserved degrees of freedom and quantifying uncertainty in string theory will also be explored elsewhere.

The aim of this paper is to determine how the association of the Born-Infeld equations for nonlinear electromagnetism with their augmented hydrodynamic counterparts discussed in \cite{Br2002} will inform us about how to add noise to the evolution of the Born-Infeld displacement flux and magnetic flux. The approach will rely on via the variational and Hamiltonian method introduced for hydrodynamics in \cite{Holm2015}, which adds noise geometrically, by a canonical transformation.  In particular, stochasticity will be introduced by making the electromagnetic flux fields $\bD$ and $\bB$ evolve under a Stratonovich stochastic flow, to be transported by stochastic vector fields carrying spatial statistical correlation information, via the Lie-derivative operation of vector fields on 2-forms $\bD$ and $\bB$. The transport obtained via the action of these spatially correlated stochastic vector fields will be implemented as a canonical transformation generated by a familiar momentum map from the electromagnetic field variables to the hydrodynamics variables; namely, the Poynting vector, $\bP:=\bD\times \bB$. 

This paper compares the effects on the equations of motion of introducing stochasticity as cylindrical Stratonovich noise \cite{Bi1981} into the Hamiltonian formulations of either Euler's equations for fluid vorticity, or the Born-Infeld electromagnetic field equations, \cite{BoIn1934}. This may seem like an unlikely comparison. However, because of an intriguing hydrodynamic analogue for electromagnetic waves, the comparison turns out to be closer than one might have thought at first glance. 

The deterministic Euler's equations for fluid vorticity, $\bsomega={\rm curl}\bu$, with divergence free Eulerian fluid velocity, $\bu$, are given in three dimensions (3D) by
\begin{align}
(\partial_t + \mathcal{L} _{u}) (\bsomega \cdot d\bS) 
= \Big(\partial_t \bsomega - \mathrm{curl}
\,(\bu\times\bsomega)\Big) \cdot d\bS = 0\,,
\quad\hbox{with}\quad
{\rm div}\bu = 0
\,.
\label{Euler-vort-eqn-2form-deter}%
\end{align}
Here, $\mathcal{L} _{u}$ denotes the Lie derivative with respect to Eulerian fluid velocity vector field $u=\bu\cdot\nabla$. Equation \eqref{Euler-vort-eqn-2form-deter} is interpreted as the familiar Lie transport of the vorticity flux, $\bsomega \cdot {\rm d}\bS = {\rm d}(\bu\cdot d\bx)$, an exact 2-form, moving with the fluid, carried by its corresponding velocity field, $\bu={\rm curl}^{-1}\bsomega$.

The Born-Infeld electromagnetic field equations were introduced in \cite{BoIn1934} as,
\begin{align}
\partial_t \bD = {\rm curl}\,\bH
\,,\quad
\partial_t \bB = -\,{\rm curl}\,\bE
\quad\hbox{with}\quad
{\rm div}\bD = 0 = {\rm div}\bB
\,,
\label{BIE-set-deter}
\end{align}
for electromagnetic fields $\bD$, $\bE$, $\bB$, $\bH$. These equations may be interpreted in the classical sense of integrals of space-time dependent 2-forms over fixed spatial domains, as
\begin{align}
\partial_t (\bD\cdot d\bS) &= {\rm d}(\bH\cdot d\bx)
\,,\quad
\partial_t (\bB\cdot d\bS) = -\,{\rm d}(\bE\cdot d\bx)
\label{BIE-set-deter-int}
\end{align}
with closed 2-forms
\begin{align}
{\rm d}(\bD\cdot d\bS) &={\rm div}\bD\,d^3x = 0 
= {\rm d}(\bB\cdot d\bS) = {\rm div}\bB\,d^3x
\,,
\label{BIE-2forms}
\end{align}
where ${\rm d}$ denotes the spatial differential (exterior derivative).
The formulas in \eqref{BIE-set-deter-int} are familiar from the weak-field limit of the Born-Infeld equations, which yields the source-free Maxwell equations, in the absence of free electrical charges and currents. In particular, the surface elements and line elements, $d\bS$ and $d\bx$, respectively,  in equations \eqref{BIE-set-deter-int} are \emph{fixed in space}; while, in contrast, $d\bS$ and $d\bx$ may be interpreted as  \emph{moving with the fluid} in equations \eqref{Euler-vort-eqn-2form-deter}, because of the Lie derivative operation, $\mathcal{L} _{u}$. 

Comparisons between Euler's fluid equations and Maxwell's field equations have been an intriguing issue in the physics literature ever since the mid-19th century. For a recent historical survey of these comparisons, see \cite{Si1991}. For further mathematical relations between hydrodynamics and the Born-Infeld model, see \cite{Br2004} and \cite{BrYo2005}.  For an in-depth, special relativistic treatment which includes interactions of Maxwell fields with fluid dynamics, see \cite{Holm1987}. 

The present paper will offer yet another aspect of these comparisons, by investigating how the introduction of stochasticity, representing various types of uncertainty, will affect the evolutionary operators in the two sets of equations. 

Euler's fluid equations and Maxwell's field equations both arise via Hamiltonian reduction by symmetry. On one hand, Eulerian fluid dynamics possesses relabelling symmetry, which allows reduction by symmetry in transforming from Lagrangian to Eulerian fluid variables. Namely, the Eulerian fluid variables are invariant under relabelling of the Lagrangian fluid particles. On the other hand, classical electromagnetic theory possesses gauge symmetry, which allows reduction by transforming from the potentials to the fields, the latter being invariant under gauge transformations of the potentials (Weyl symmetry).  The Hamiltonian structures resulting from these two types of symmetry reduction are quite different. However, they each result in a map from canonical field variables to a momentum variable taking values in the space of 1-form densities, dual to vector fields with respect to $L^2$ pairing. This property of sharing a \emph{momentum map} from the canonical field variables to a momentum density will provide an avenue for introducing a stochastic vector field into both models by using their shared Hamiltonian structure.

{\bf Plan.} In the remainder of the paper, section \ref{sec-fluids} sketches the method of \cite{Holm2015} for introducing  Stratonovich noise into Hamiltonian dynamics of nonlinear field theories by using momentum maps dual to vector fields. This approach is illustrated by comparing its results for two apparently different theories; namely, ideal fluid dynamics in \ref{sec-fluids} and Born-Infeld electromagnetism in section \ref{sec-ABI}. The main part of the paper concludes and summarise the results in section \ref{sec-conclude-sum}. Lagrangian and Hamiltonian variational formulations of ABI are provided in Appendix \ref{sec-var-form}. The corresponding results for the high field MHD limit of the ABI equations are discussed in the Appendix \ref{sec-MHD-form}.  

 The fundamentals of the Hamiltonian structures for the two theories are reviewed from first principles for fluid dynamics in \cite{HoMaRa1998} and for the Born-Infeld theory in Appendices \ref{sec-var-form} and \ref{sec-MHD-form}. 
These Appendices derive the connection between the Hamiltonian structures for the Born-Infeld  equations (canonical Poisson bracket) and ideal hydrodynamics (Lie-Poisson bracket). Appendix \ref{sec-var-form} does this in general, and Appendix \ref{sec-MHD-form} discusses the high-field limit. A connection to magnetohydrodynamics is also revealed in the high-field limit discussed in Appendix \ref{sec-MHD-form}. 

The Poynting vector momentum map is the key to understanding the analogy between electromagnetism and hydrodynamics. To augment the Hamiltonian operator for Born-Infeld electromagnetism to include the Poynting vector, we follow a mathematical approach introduced in Krishnaprasad and Marsden \cite{KrMa1987} for deriving the dynamics of a rigid body with flexible attachments. This approach leads to a compound Poisson structure that may be written as the sum of a canonical structure and Lie-Poisson structure obtained from a cotangent-lift momentum map in which variations are applied independently. This augmented Poisson bracket for the Born-Infeld theory provides a fundamental explanation of the hydrodynamic analogy for electromagnetism. Namely, the cotangent-lift momentum map to the Poynting vector, corresponding to Lie transformations of the  canonical electromagnetic variables by smooth vector fields, leads via  \cite{KrMa1987} to an augmented Poisson bracket for electromagnetism which satisfies the same semidirect-product Lie-Poisson bracket relations as those found in ideal continuum dynamics \cite{HoMaRa1998}.

\section{Stochastic fluid dynamics} \label{sec-fluids}
A variational approach to stochastic fluid dynamics has recently been derived in \cite{Holm2015} and its remarkable analytical properties have been investigated in \cite{CrFlHo2017} for the particular case of the 3D stochastic Euler fluid equation, given in terms of the stochastic time derivative $\dd$ by 
\begin{align}
0 = (\dd + \mathcal{L} _{\dd{\by}_{t}}) (\boldsymbol{\omega} \cdot
d\boldsymbol{S}) 
= \Big(\dd \boldsymbol{\omega} - \mathrm{curl}
\,(\dd{\by}_{t}\times\boldsymbol{\omega})\Big) \cdot d\boldsymbol{S} \,,
\label{Euler-vort-eqn-2form-stoch}%
\end{align}
with the \textit{Stratonovich} stochastic, divergence-free vector field,
\begin{align}
\dd{\by}_{t} = \bu({\by}_{t},t)dt + \sum_{i} \bsxi_{i}({\by}_{t}) \circ dW^{i}_{t}
\,, \label{StochVF}%
\end{align}
in which each of the summands has zero divergence.
Here, $\dd{\by}_{t}$ denotes the stochastic process and the second term in \eqref{StochVF} constitutes cylindrical Stratonovich noise, in which the amplitude of the noise depends on space, but not time. An immediate consequence of the stochastic fluid equation in \eqref{Euler-vort-eqn-2form-stoch} is a stochastic version of the Kelvin circulation theorem, so that 
\begin{align}
\dd \!\!\int_{c(\dd{\by}_{t} )} \bu(\bx,t) \cdot d\bx = 0
\,, \label{Stoch-Kel-Thm}%
\end{align}
for any fluid material loop $c(\dd{\by}_{t} )$ moving with the Stratonovich stochastic vector field $\dd{\by}_{t}$ in \eqref{StochVF}. 

For in-depth treatments of cylindrical noise, see \cite{Pa2007,Sc1988}. In our case, the $\bsxi_{i}({\by}_{t})$, $i=1,2,\dots,N$, appearing in the stochastic vector field in \eqref{StochVF} comprise $N$ prescribed, time independent, divergence-free vectors which ideally may be obtained from data measured at fixed points $\bx$ along the Lagrangian path ${\by}_{t}$. For example, one may take  the $\bsxi_{i}(\bx)$ to be Emperical Orthogonal Functions (EOFs), which are eigenvectors of the velocity-velocity correlation tensor for a certain measured flow with stationary statistics \cite{HaJoSt2007}. The $\bsxi_{i}(\bx)$ may also be obtained numerically by comparisons of Lagrangian trajectories at fine and coarse space and time scales \cite{CoCrHoShWe2017}.

It may not be surprising that the variational introduction of cylindrical Stratonovich noise into Euler's fluid equation proposed  in \cite{Holm2015} for fluids has simply introduced an additional, stochastic vector field $\sum_{i} \bsxi_{i}(\bx) \circ dW^{i}_{t}$ into equation \eqref{StochVF} which augments the Lie transport in equation \eqref{Euler-vort-eqn-2form-stoch} in the Eulerian representation, while preserving its Hamiltonian geometric structure and many of its analytical properties. After all, the essence of Euler fluid dynamics is Lie transport \cite{HoMaRa1998}.  However, it might be more surprising if the variational introduction of noise into the Born-Infeld electromagnetic field equations turned out to introduce the same sort of stochastic Lie transport, for example, in the displacement current. Investigating this issue and explaining it will be our concern for the remainder of the paper. 

\section{Augmented Born-Infeld (ABI) equations} \label{sec-ABI}

\subsection{Deterministic Born-Infeld equations}
Below, we will introduce stochasticity into the deterministic Born-Infeld equations, which may be written in Hamiltonian form in the rest frame as \cite{BoIn1934}
\begin{align}
\begin{split}
\begin{bmatrix}  
\partial_t \bD
\\ 
\partial_t \bB
\end{bmatrix}
=
\begin{bmatrix}
\{ \bD\,,\, \bbH(\bD,\bB)\,\}
\\ 
\{ \bB\,,\, \bbH(\bD,\bB)\,\}
\end{bmatrix}
=
\begin{bmatrix}
0   & {\rm curl}
\\ 
-\,{\rm curl} & 0
\end{bmatrix}
\begin{bmatrix}
{\delta \bbH}/{\delta \bD}
\\ 
{\delta \bbH}/{\delta \bB}
\end{bmatrix}
\quad\hbox{with}\quad
{\rm div}\bD = 0 = {\rm div}\bB
\,.
\end{split}
\label{BI-eqns-Ham}
\end{align}
The divergence free conditions on $\bD$ and $\bB$ in \eqref{BI-eqns-Ham} continue to hold, provided they hold initially. As discussed in \cite{MaRa1994}, the Poisson bracket in \eqref{BI-eqns-Ham} was initially due to \cite{Pauli1933} and it may be written equivalently as 
\begin{align}
\Big\{F,K \Big\}(\bD,\bB)
=
\int \frac{\delta F}{\delta \bD}\cdot {\rm curl}\frac{\delta K}{\delta \bB}
- \frac{\delta K}{\delta \bD}\cdot {\rm curl}\frac{\delta F}{\delta \bB}
\,d^3x\,.
\label{BI-eqns-PB-B}
\end{align}
The Poisson bracket \eqref{BI-eqns-PB-B} is equivalent to the canonical Poisson bracket in terms of the magnetic vector potential, $\bA$  and (minus) the displacement vector, $-\bD$. 
\begin{align}
\Big\{F,K \Big\}(\bD,\bA)
=
\int \frac{\delta F}{\delta \bD}\cdot \frac{\delta K}{\delta \bA}
- \frac{\delta K}{\delta \bD}\cdot \frac{\delta F}{\delta \bA}
\,d^3x\,,
\label{BI-eqns-PB-A}
\end{align}
under the change of variable $\bB={\rm curl}\bA$. Geometrically, the magnetic vector potential, $\bA$ 
(resp. displacement vector, $\bD$) defines the components of a 1-form $A=\bA\cdot d\bx$ (resp. a 2-form $D=\bD\cdot d\bS$).
For more details about the Born-Infeld Hamiltonian structure, see Appendix \ref{sec-var-form}. 
See also \cite{Di1960} for a manifestly Lorentz invariant formulation of the Born-Infeld field theory. 

In terms of the electromagnetic fields, $(\bD,\bB)$, the Born-Infeld Hamiltonian $\bbH(\bD,\bB)$ is given by
\begin{align}
\bbH (\bD,\bB) = \int \cH(\bD,\bB) \,d^3x
\quad\hbox{where}\quad
\cH(\bD,\bB) = \sqrt{1 + |\bD|^2 + |\bB|^2 + |\bD\times \bB|^2} 
\,.
\label{BI-Ham}
\end{align}
The variational derivatives of the Hamiltonian $\bbH$ required in \eqref{BI-eqns-Ham} are given by
\begin{align}
\bE = \frac{\delta \bbH}{\delta \bD} = \frac{\bD + \bB\times \bP}{\cH}
\quad\hbox{and}\quad
\bH = \frac{\delta \bbH}{\delta \bB} = \frac{\bB - \bD\times \bP}{\cH}
\,,\label{var-deriv-Ham}
\end{align}
where $\bP$ is the Poynting vector, given by
\begin{align}
\bP  = \bD\times \bB = \bE\times \bH
\,.
\label{P-def}
\end{align}
Thus, the Born-Infeld equations comprise four equations in fixed Eulerian coordinates,
\begin{align}
\partial_t \bD = {\rm curl}\,\bH
\,,\quad
\partial_t \bB = -\,{\rm curl}\,\bE
\quad\hbox{with}\quad
{\rm div}\bD = 0 = {\rm div}\bB
\,.
\label{BIE-set}
\end{align}

\begin{remark}[Conservation laws, \cite{Br2004}]
The dynamical equations for the energy density $\cH$ and the Poynting vector $\bP$ (momentum density, also energy flux density) may be written in conservative form, as
\begin{align}
\begin{split}
&\partial_t \cH + {\rm div}\,\bP = 0
\,,\\&
\partial_t \bP + {\rm div}\,
\left( \frac{\bP\otimes\bP}{\cH} - \frac{\bD\otimes\bD}{\cH} - \frac{\bB\otimes\bB}{\cH}
\right)
= \nabla \left( \frac{1}{\cH}\right) 
\,.
\end{split}
\label{cH+bP-cons}
\end{align}
Respectively, these relations imply conservation laws for the Born-Infeld total energy $\int\cH \,d^3x$ and total momentum $\int\bP \,d^3x$, arising due to Noether symmetries under time and space translation invariance of the Born-Infeld Hamiltonian, $\bbH(\bD,\bB)$, in \eqref{BI-Ham}. The union of the sets of the Born-Infeld equations \eqref{BIE-set} and the local conservation laws \eqref{cH+bP-cons} is called the \emph{Augmented Born-Infeld (ABI) equations} in \cite{Br2004,BrYo2005}.
\end{remark}
\begin{remark}[Hydrodynamic analogy]
From their equations in \eqref{cH+bP-cons} we see that $\cH$ is a scalar density, while $\bP$ is a 1-form density.
Thus, we may write these equations in a more geometric form, reminiscent of hydrodynamics, upon introducing the  vector notation $\bv=\bP/\cH$, $\bsgamma := \bD/\cH$ and $\bsbeta := \bB/\cH$. Equations \eqref{cH+bP-cons} then may be written equivalently as
\begin{align}
\begin{split}
(\partial_t + \cL_v)(\cH\,d^3x) &= 0
\,,\\
(\partial_t + \cL_v)(\bv\cdot d\bx ) 
- \cL_\gamma (\bsgamma\cdot d\bx)
- \cL_\beta (\bsbeta\cdot d\bx)
&=
\frac12 d\left( \cH^{-2}  + |\bv|^2 - |\bsgamma|^2   - |\bsbeta|^2  \right)
\,.\end{split}
\label{cH+bP-geom}
\end{align}
Here $\bv=\bP/\cH$ is a velocity vector and $\cL_v(\bv\cdot d\bx )$ is the Lie derivative of the 1-form $(\bv\cdot d\bx )$ with respect to the vector field $v$, whose components are given by $v=\bv\cdot\nabla=v^j\partial_j$. Namely,
\begin{align}
\cL_v(\bv\cdot d\bx ) = \big((\bv\cdot \nabla)\bv + v_j\nabla v^j\big)\cdot d\bx
= \big( -\bv \times {\rm curl}\,\bv + \nabla (|\bv|^2) \big)\cdot d\bx\,.
\label{Lie-deriv-1form}
\end{align}
Formulas analogous to \eqref{Lie-deriv-1form} also exist for the vectors $\bsgamma := \bD/\cH$ and $\bsbeta := \bB/\cH$, in computing Lie derivatives with respect to vector fields $\gamma=\bsgamma\cdot\nabla$ and $\beta=\bsbeta\cdot\nabla$ applied to the 1-forms $(\bsgamma\cdot d\bx)$ and $(\bsbeta\cdot d\bx)$, respectively.

Thus, the Born-Infeld evolution equations for $\cH$ and $\bP$ in geometric form \eqref{cH+bP-geom} are analogous to similar equations in ideal fluid dynamics. Following this hydrodynamic analogy for the Born-Infeld equations, the corresponding Born-Infeld Kelvin circulation theorem may be found by integrating equation \eqref{cH+bP-geom} around a closed loop $c(v)$ moving with the velocity $\bv(\bx,t)$, to obtain 
\begin{align}
\frac{d}{dt}\oint_{c(v)} (\bv\cdot d\bx ) 
+ \oint_{c(v)} \big(\bsgamma\times {\rm curl}\bsgamma 
+  \bsbeta\times {\rm curl}\bsbeta \big) \cdot d\bx
= 0\,,
\label{BI-KelThm}
\end{align}
where we have used the fundamental theorem of calculus $\oint d\alpha = 0$ for the loop integral of the differential of any scalar function $\alpha$ to evaluate the right hand side. The Born-Infeld Kelvin circulation theorem \eqref{BI-KelThm} implies that the circulation of $\bv=\bP/\cH$ around a loop moving with the flow of $\bv$ will in general not be preserved, unless the sum of the cross products of the velocities $\bsgamma=\bD/\cH$ and $\bsbeta=\bB/\cH$ with their respective curls is proportional to the gradient of a scalar function. 

Yet another fluidic analogy to the vorticity equation \eqref{Euler-vort-eqn-2form-deter} may be obtained in terms of a vorticity $\bsvarpi:={\rm curl}\bv$ by applying Stokes theorem to equation \eqref{BI-KelThm}, to find
\begin{align}
(\partial_t + \mathcal{L} _{v}) (\bsvarpi \cdot d\bS) 
= \Big(\partial_t \bsvarpi - \mathrm{curl}
\,(\bv\times\bsvarpi)\Big) \cdot d\bS 
= 
-\,{\rm curl}\Big( \bsgamma\times {\rm curl}\bsgamma + \bsbeta\times {\rm curl}\bsbeta \Big) \cdot d\bS 
\,.
\label{BI-vort-eqn-2form-deter}%
\end{align}

\end{remark} 

\begin{remark} 
The Born-Infeld equations \eqref{BI-eqns-Ham} comprise a nonlinear deformation of Maxwell's equations. Indeed, Maxwell's equations may be recovered from the variational equations \eqref{var-deriv-Ham} for appropriately small magnitudes $|\bD|\ll1$ and $|\bB|\ll1$, for which ${\delta \bbH} /{\delta \bD} \to \bD$ and ${\delta \bbH}/{\delta \bB}\to\bB$. Conservation equations analogous to \eqref{cH+bP-cons} and  circulation equations analogous to \eqref{BI-KelThm} and \eqref{BI-vort-eqn-2form-deter} also hold for the Maxwell case. \end{remark}

To provide a geometric motivation for the fluid dynamics interpretation via the circulation theorem \eqref{BI-KelThm} for the deterministic augmented Born-Infeld equations in \eqref{cH+bP-cons}, we introduce a bit of standard terminology from geometric mechanics. 

\begin{definition}[Cotangent lift momentum map]
Suppose $G$ is a Lie group which acts on a configuration manifold $Q$ and, hence, on its canonical phase space $T^*Q$ by cotangent lifts. The corresponding momentum map $J(q,p)$ from the canonical phase space $T^*Q$ to the dual $\mathfrak{g}^*$ of the Lie algebra $\mathfrak{g}$ of Lie group $G$ is given by \cite{MaRa1994,Ho2011}
\begin{align}
\big\langle J(q,p)\,,\,\xi\big\rangle_{\mathfrak{g}}
= \big\langle\big\langle p_q\,,\,\Phi_\xi (q) \big\rangle\big\rangle_{TQ}
\,,
\label{CLmomap-def}
\end{align}
for $(q,p)\in T^*Q$, $\xi\in \mathfrak{g}$, $J(q,p)\in \mathfrak{g}^*$, $p_q$ the canonical phase space momentum at position $q$, $\Phi_\xi (q)$ the infinitesimal transformation of $q$ by $G$, and natural pairings $\langle\,\cdot\,,\cdot\,\rangle_{\mathfrak{g}}: \mathfrak{g}^*\times \mathfrak{g} \to \mathbb{R}$ and $\langle\langle\,\cdot\,,\cdot\,\rangle\rangle_{TQ}: T^*Q\times TQ \to \mathbb{R}$.
\end{definition}

In the Born-Infeld case, we replace the Lie group $G$ in the definition above by the diffeomorphisms ${\rm Diff}(\mathbb{R}^3)$. We then define the canonical phase space $T^*\mathcal{D}$ as the set of pairs $(\bD,\bA)$ whose canonical Poisson bracket is given in \eqref{BI-eqns-PB-A} and take the infinitesimal transformation $\Phi_\xi (q)$ to be $-\mathcal{L}_\xi D$; namely, (minus) the Lie derivative of the closed 2-forms $D=(\bD\cdot d\bS)$ by the divergence free vector fields, $\xi:=\bsxi\cdot\nabla\in\mathfrak{X}(\mathbb{R}^3)$ with ${\rm div}\bsxi(\bx)=0$. In terms of these variables, we find the following. 

\begin{theorem}
The 1-form density $P:=\bP\cdot d\bx \otimes d^3x$ with components given by the Poynting vector $\bP=\bD\times\bB$ defines a cotangent lift momentum map, $T^*\mathcal{D} \to \mathfrak{X}^*$, from the canonical phase space $T^*\mathcal{D}$ identified with the set of pairs $(\bD,\bA)$ whose canonical Poisson bracket is given in \eqref{BI-eqns-Ham}, to the dual space $\mathfrak{X}^*(\mathbb{R}^3)$ of the smooth vector fields $\mathfrak{X}(\mathbb{R}^3)$ with respect to the $L^2$ pairing. 
\end{theorem}

\begin{proof}
Referring to the definition of cotangent lift momentum map in \eqref{CLmomap-def}, we compute
\begin{align}
\begin{split}
\big\langle P,\xi\big\rangle_{\mathfrak{X}}\, &= \int P ({\xi}) 
:= \int \bsxi\cdot\bP\,d^3x
\\&= \int \bB\cdot \bsxi\times \bD\,d^3x
= \int \bA\cdot {\rm curl}(\bsxi\times \bD)\,d^3x
\\&= -\int \bA\cdot d\bx\wedge \mathcal{L}_\xi  (\bD\cdot d\bS)
\\&= \big\langle\big\langle A\,,\,-\mathcal{L}_\xi D \big\rangle\big\rangle
\\&=: \big\langle A\diamond D,\xi\big\rangle_{\mathfrak{X}}
\,,
\end{split}
\label{P-xi-pairing}
\end{align}
where we denote $\xi:=\bsxi\cdot\nabla$, $A=\bA\cdot d\bx$ and $D=\bD\cdot d\bS$ in $\mathbb{R}^3$ coordinates. In the last line of \eqref{P-xi-pairing}, the diamond $(\,\diamond\,)$ operation is defined; see, e.g., \cite{HoMaRa1998}. We also take homogeneous spatial boundary conditions, whenever we integrate by parts.
\end{proof}

\begin{remark}
Although cotangent lift momentum maps are known to be equivariant \cite{MaRa1994}, we may check directly that the momentum map $T^*\mathcal D \to \mathfrak{X}^*$ is infinitesimally equivariant by showing that the map is Poisson for the proper Lie-Poisson bracket. To do this, we use the canonical Poisson brackets in \eqref{BI-eqns-PB-A} for pairs $(\bD,\bA)$ and apply the chain rule to compute the Poisson brackets  
$\big\{P_i(\mathbf{x}),P_j(\mathbf{y})\big\}$ for the $\mathbb{R}^3$ components of the Poynting vector $\bP$. After a direct calculation by change of variables, we find
\begin{eqnarray}\label{momentum-map-bracket1}
\big\{P_i(\mathbf{x})\,,\,P_j(\mathbf{y})\big\}
=
-\,
\Big(
P_j(\mathbf{x})\frac{\partial}{\partial x^i}
+
\frac{\partial}{\partial x\,^j}\,P_i(\mathbf{x})\Big)
\delta(\mathbf{x}-\mathbf{y})
\,,\label{LP-bracket}
\end{eqnarray}
in which the right hand side is the Hamilton operator for the Lie-Poisson bracket on the space of 1-form densities. This calculation proves directly that the map $(\bD,\bA)\in T^*\mathcal{D} \to P = \bP\cdot d\bx \otimes d^3x\in\mathfrak{X}^*$ is an infinitesimally equivariant momentum map. Importantly for our fluid-fields analogy, the Lie-Poisson bracket \eqref{LP-bracket} for the electromagnetic field momentum has the same form as the Lie-Poisson bracket for fluid momentum and both of these are cotangent lift momentum maps of the same type, see \cite{Holm1987,Holm2015, HoMaRa1998} and references therein. In fact, we may now write the local conservation law for Poynting vector in \eqref{cH+bP-cons} more geometrically as 
\begin{align}
\partial_t P + \big\{P\,,\,\mathbb{H}\big\}  
= \partial_t P + \pounds_{\delta \mathbb{H}/\delta P}P  
= - \frac{\de \mathbb{H}}{\de B}\diamond B 
- \frac{\de \mathbb{H}}{\de D}\diamond D 
\,,\label{LP-form}
\end{align}
where $\pounds_{\delta \mathbb{H}/\delta P}P$ is the Lie derivative of the 1-form density $P\in\mathfrak{X}^*$ by the vector field $v=\de \mathbb{H}/\delta P$; the fluxes $B=\bB\cdot d\bS$ and $D=\bD\cdot d\bS$ are 2-forms; $\de H/\de B$ and  $\de H/\de D$ are 1-forms; the diamond $(\diamond)$ operator is defined in \eqref{P-xi-pairing}; $\mathbb{H} $ is the Born-Infeld Hamiltonian in \eqref{BI-Ham} now written in terms of vectors $(\bB,\bD,\bP)$, as
\begin{align}
\bbH (\bB,\bD,\bP) = \int \cH(\bB,\bD,\bP) \,d^3x
\quad\hbox{where}\quad
\cH(\bB,\bD,\bP) = \sqrt{1 + |\bD|^2 + |\bB|^2 + |\bP|^2} 
\,,
\label{BI-Ham-P}
\end{align}
and one treats the variational derivatives in $\bB$,  $\bD$ and $\bP$ appearing in \eqref{LP-form} as being independent; that is, ${\delta \bbH }/{\delta \bB}=(\bB/\cH)$, ${\delta \bbH }/{\delta \bD}=(\bD/\cH)$, and ${\delta \bbH }/{\delta \bP}=(\bP/\cH)=\bv$. 
\end{remark}

\subsection{Stochastic Born-Infeld equations}

Following \cite{Bi1981}, one may introduce stochasticity into the deterministic Born-Infeld equations in their Hamiltonian form \eqref{BI-eqns-Ham} by adding a stochastic term to the Hamiltonian density and using the same Poisson bracket structure as in \eqref{BI-eqns-Ham} for the deterministic case. This is done via the replacements
\begin{align}
\partial_t \to \dd 
\quad\hbox{and}\quad
\cH(\bD,\bB) \to \cH(\bD,\bB)\,dt + \sum_i h_i(\bD,\bB)\circ dW^i_t
\,,
\label{cH+bP-stoch}
\end{align}
where $\dd$ denotes stochastic time derivative, and $\cH(\bD,\bB)$ is the Hamiltonian for the deterministic equations in \eqref{BI-Ham}. The stochastic term, $\sum_i h_i(\bD,\bB)\circ dW^i_t$, in the perturbed Hamiltonian in formula \eqref{cH+bP-stoch} comprises cylindrical Stratonovich noise, written as a sum over $N$ Brownian motions $dW^i_t$, $i=1,2,\dots,N$, each interacting with the dynamical variables $(\bD,\bB)$ via its own Hamiltonian amplitude, $h_i(\bD,\bB)$, \cite{Bi1981}. 

The wide range of potential choices for the stochastic Hamiltonians $h_i(\bD,\bB)$ in \eqref{cH+bP-stoch} may be narrowed considerably by interpreting them as Hamiltonian densities that couple the noise to the drift term based on the deterministic Hamiltonian density. Specifically, we shall choose to include the Stratonovich noise in \eqref{cH+bP-stoch} by coupling it with the momentum density, as done in \cite{Holm2015} for stochastic fluid dynamics. With this choice, written explicitly below in \eqref{h-cyl-noise}, one interprets the motion induced by the total Hamiltonian density in \eqref{cH+bP-stoch} as the sum of the drift part of the stochastic motion, as governed by the deterministic Hamiltonian, plus a stochastic Hamiltonian perturbation, generated by a process of state-dependent random displacements. In the Born-Infeld case, the Poynting vector $\bP=\bD \times \bB$ is the momentum density, according to \eqref{cH+bP-cons}. Because $\bD$ and $\bB$  lie in the polarisation plane normal to the direction of propagation for electro-magnetism, we know that the Poynting vector $\bP$ lies along the direction of propagation. To proceed, we choose the stochastic part of the Born-Infeld Hamiltonian density in \eqref{cH+bP-stoch} to be
\begin{align}
h_i(\bD,\bB) = \bsxi_i(\mathbf{x}) \cdot \bD \times \bB
\,,\label{h-cyl-noise}
\end{align}
for a set of prescribed divergence-free vector fields $\bsxi_i(\mathbf{x})$ which are meant to represent the stationary spatial correlations of the cylindrical noise. 
\begin{remark}
Notice that the inclusion of the stochastic term \eqref{h-cyl-noise} into the Hamiltonian density in \eqref{cH+bP-stoch} has introduced explicit dependence on time and space coordinates in the total Hamiltonian. Consequently, the conservation laws in \eqref{cH+bP-cons}  for the deterministic Born-Infeld total energy $\int\cH \,d^3x$ and total momentum $\int\bP \,d^3x$ may no longer apply in the stochastic case; see, however, Remark \ref{cons-erg-mom}. Moreover, the stochastic Born-Infeld equations in \eqref{BI-eqns-stoch} are no longer Lorentz invariant, although this was to be expected, because loss of explicit Lorentz invariance arises, in general, when casting Lorentz invariant dynamics into the Hamiltonian formalism. 
 \end{remark}
\begin{remark} 
The stochastic part of the Born-Infeld Hamiltonian density in \eqref{h-cyl-noise} is the integrand in the first line of \eqref{P-xi-pairing}. Thus, the Stratonovich noise in \eqref{cH+bP-stoch} has been coupled to the deterministic Born-Infeld field theory through the momentum map in \eqref{P-xi-pairing} corresponding to the Poynting vector, $P=A\diamond D = \bD \times \bB\cdot d\bx \otimes d^3x$, which is a 1-form density. 
 \end{remark}

\noindent
\textbf{Stratonovich form.} Upon performing the indicated operations in \eqref{BI-eqns-Ham}, \eqref{cH+bP-stoch} and \eqref{h-cyl-noise}, one finds the following set of stochastic Born-Infeld equations. 
\begin{align}
\begin{split}
\begin{bmatrix}
\dd \bD
\\ \\
\dd \bB
\end{bmatrix}
=
\begin{bmatrix}
{\rm curl}\,\bH \,d t - \sum_i [\,\xi_i,\bD\,]\circ dW^i_t
\\ \\
-\, {\rm curl}\,\bE \,d t - \sum_i [\,\xi_i,\bB\,]\circ dW^i_t
\end{bmatrix}
\quad\hbox{with}\quad
{\rm div}\bD = 0 = {\rm div}\bB
\,,
\end{split}
\label{BI-eqns-stoch}
\end{align}
where $[\,\cdot\,,\,\cdot\,]$ denotes the Lie bracket of divergence-free vector fields, defined by
\begin{align}
[\,\xi_i\,,\bD\,] = -\,{\rm curl}(\bsxi_i\times \bD)
= (\bsxi_i \cdot \nabla)\bD - (\bD \cdot \nabla)\bsxi_i
 =:  \cL_{\xi_i}\bD
\,,
\label{LieBrkt-def}
\end{align}
with Lie derivative $\cL_{\xi_i}$ with respect to $\xi_i$, and the vectors $\bE$ and $\bH$ in \eqref{BI-eqns-stoch} are defined via the variational derivatives in \eqref{var-deriv-Ham}. Equations \eqref{BI-eqns-stoch} represent the stochastic Born-Infeld versions of the displacement current and the flux rule for Maxwell's equations. Upon rewriting equations \eqref{BI-eqns-stoch} as 
\begin{align}
\begin{split}
\begin{bmatrix}
\dd \bD
+ \sum_i \cL_{\xi_i}\bD \circ dW^i_t
\\ \\
\dd \bB
+ \sum_i \cL_{\xi_i}\bB \circ dW^i_t
\end{bmatrix}
=
\begin{bmatrix}
{\rm curl}\,\bH \,d t 
\\ \\
-\, {\rm curl}\,\bE \,d t 
\end{bmatrix}
\quad\hbox{with}\quad
{\rm div}\bD = 0 = {\rm div}\bB
\,,
\end{split}
\label{BI-eqns-stoch-Lie}
\end{align}
one sees on the left side of \eqref{BI-eqns-stoch-Lie} that the present Hamiltonian approach to adding stochasticity to the Born-Infeld field equations has introduced Lie transport by Stratonovich noise into the stochastic time derivative, exactly as it did in fluid dynamics, treated in \cite{Holm2015}. This happened because we coupled the noise to the momentum map -- the Poynting co-vector density -- which generates infinitesimal spatial translations under the Lie-Poisson bracket in \eqref{LP-bracket}. For more discussion of this point, see Appendix \ref{sec-var-form}. 

To pay a bit more attention to the differential geometry of this problem, we rewrite the stochastic Born-Infeld equations \eqref{BI-eqns-stoch} as evolutions of 2-forms,
\begin{align}
\begin{split}
\Big( \dd + \sum_i \cL_{\xi_i}(\,\cdot\,) \circ dW^i_t \Big)
\begin{bmatrix}
 \bD\cdot d\bS
\\ 
\bB\cdot d\bS
\end{bmatrix}
=
{\rm d}
\begin{bmatrix}
\bH\cdot d\bx \,d t 
\\ \\
-\, \bE\cdot d\bx \,d t 
\end{bmatrix}
\quad\hbox{with}\quad
{\rm div}\bD = 0 = {\rm div}\bB
\,,
\end{split}
\label{BI-eqns-stoch-geom-int}
\end{align}
where ${\rm d}$ is the spatial differential (exterior derivative).

In integral form, this is equivalent to
\begin{align}
\begin{split}
\int_S
\Big( \dd + \sum_i \cL_{\xi_i}(\,\cdot\,) \circ dW^i_t \Big)
\begin{bmatrix}
 \bD\cdot d\bS
\\ 
\bB\cdot d\bS
\end{bmatrix}
=
\oint_{\partial S}
\begin{bmatrix}
\bH\cdot d\bx \,d t 
\\ \\
-\, \bE\cdot d\bx \,d t 
\end{bmatrix}
\quad\hbox{with}\quad
{\rm div}\bD = 0 = {\rm div}\bB
\,,
\end{split}
\label{BI-eqns-stoch-geom}
\end{align}
in which one sees that the fluxes of $\bD$ and $\bB$ are frozen into the  displacements generated by the Lie derivative with respect to the stochastic vector field $\sum_i \cL_{\xi_i}(\,\cdot\,) \circ dW^i_t $.
To reiterate, the effect of introducing Stratonovich noise coupled to the momentum map while preserving the Poisson structure of a Hamiltonian system is to introduce stochastic Lie transport into the evolution operator, thereby introducing random time-dependent uncertainty into the differential surface elements, while preserving the remainder of the unperturbed equations, precisely as found earlier for stochastic fluid equations in \cite{Holm2015}.

\medskip

\noindent
\textbf{It\^o form.} 
When dealing with cylindrical noise, the spatial coordinates are treated merely as parameters. That is, one may regard the cylindrical noise process as a finite dimensional stochastic process parametrized by $\bx$ (the spatial coordinates).  In this regard, the Stratonovich equation makes analytical sense pointwise, for each fixed $\bx$.  Once this is agreed, then the transformation to It\^o by the standard method also makes sense pointwise in space. For more details, see \cite{Pa2007,Sc1988}.

The It\^o form of the stochastic Born-Infeld equations in \eqref{BI-eqns-stoch} is given by 
\begin{align}
\begin{split}
\begin{bmatrix}
\dd \bD
+ \sum_i  \cL_{\xi_i}\bD \,dW^i_t 
- \frac12 \sum_i  \cL_{\xi_i}(\cL_{\xi_i}\bD)\,dt
\\ \\
\dd \bB 
+ \sum_i  \cL_{\xi_i}\bB \,dW^i_t
- \frac12 \sum_i  \cL_{\xi_i}(\cL_{\xi_i}\bB)\,dt
\end{bmatrix}
=
\begin{bmatrix}
{\rm curl}\,\bH \,d t 
\\ \\
-\, {\rm curl}\,\bE \,d t 
\end{bmatrix}
\quad\hbox{with}\quad
{\rm div}\bD = 0 = {\rm div}\bB
\,.
\end{split}
\label{BI-eqns-stoch-Ito}
\end{align}
In It\^o form, the stochastic Born-Infeld equations contain both Lie transport by It\^o noise and an elliptic operator (double Lie derivative) arising from the It\^o contraction term. 

\begin{remark}[Application to the stochastic Maxwell equations in the weak-field limit] The reduction of the stochastic Born-Infeld equations in \eqref{BI-eqns-stoch} and \eqref{BI-eqns-stoch-Ito} to the linear stochastic Maxwell equations occurs in the weak-field limit via replacing the Born-Infeld Hamiltonian density $\cH$ in \eqref{BI-Ham} by the Maxwell energy density 
\[
\cH \to \frac12\left(|\bD|^2 + |\bB|^2\right)\,.
\]
In this weak-field limit for the Hamiltonian density, the variables defined by variational derivatives in \eqref{BI-eqns-stoch-Ito} then reduce as $\bE\to\bD$ and $\bH\to\bB$. In this limit, the nonlinear stochastic Born-Infeld equations  in \eqref{BI-eqns-stoch-Ito} reduce to the linear stochastic Maxwell equations, a variant of which has been investigated, e.g., in \cite{HoStYa2010} for its approximate controllability. See, e.g., also \cite{HoJiZh2014} for a numerical study of a simplified version of these stochastic Maxwell equations in two spatial dimensions.\medskip

Taking the expectation of the It\^o system \eqref{BI-eqns-stoch-Ito} in the Maxwell weak-field limit $\bE\to\bD$ and $\bH\to\bB$ implies that the expected values for the electric field $\langle\bD\rangle:=\mathbb{E}[\bD]$ and the magnetic field $\langle\bB\rangle:=\mathbb{E}[\bB]$ evolve according to
\begin{align}
\begin{split}
\begin{bmatrix}
\partial_t \langle\bD\rangle
\\ \\
\partial_t \langle\bB\rangle
\end{bmatrix}
=
\begin{bmatrix}
{\rm curl}\,\langle\bB\rangle 
+ \frac12 \sum_i  \cL_{\xi_i}(\cL_{\xi_i}\langle\bD\rangle)
\\ \\
-\, {\rm curl}\,\langle\bD\rangle  
+ \frac12 \sum_i  \cL_{\xi_i}(\cL_{\xi_i}\langle\bB\rangle)
\end{bmatrix}
\quad\hbox{with}\quad
{\rm div}\langle\bD\rangle = 0 = {\rm div}\langle\bB\rangle
\,.
\end{split}
\label{Maxwell-eqns-stoch-expected}
\end{align}
In the case that the $\bsxi_i$ are the constant coordinate basis vectors in $\mathbb{R}^3$, one replaces $\cL_{\xi_i}\langle\bD\rangle \to {\bsxi_i}\cdot\nabla\langle\bD\rangle$. Hence, in this case, the sums in \eqref{Maxwell-eqns-stoch-expected} over double Lie derivatives simply reduce to Laplacians.

\end{remark}

\section{Conclusion}\label{sec-conclude-sum}

In conclusion, we have seen that the association of the Born-Infeld equations \eqref{BI-eqns-Ham} with their augmented hydrodynamic counterparts in \eqref{cH+bP-geom}  has informed us via the method introduced for hydrodynamics in \cite{Holm2015} how to add noise to the evolution of the Born-Infeld displacement flux and magnetic flux; so as to preserve its energy $\cH$ in \eqref{cH+bP-geom} as well as implementing the addition of noise geometrically as a process of moving into a stochastic frame of motion.  Namely, the flux fields $\bD$ and $\bB$ evolve under a Stratonovich stochastic \emph{flow}, being Lie-transported by a sum of stochastic vector fields $\xi_i(\bx)$ carrying spatial correlation information, via the Lie-derivative operation $\sum_i \cL_{\xi_i \circ dW^i_t}(\,\cdot\,)$. 

Section \ref{sec-fluids} concluded via a variational approach to stochastic fluid dynamics that  noise enhanced Lie transport of the fluid vorticity in equation \eqref{StochVF}. For fluids, this conclusion seemed intuitive. On the other hand, it  seemed less intuitive to find in section \ref{sec-ABI} that noise would induce Lie transport of the electromagnetic fields $\bD$ and $\bB$  as 2-forms along random paths generated by the vector fields associated with the spatial correlations of the stochasticity. The correspondence between vorticity 2-forms and electromagnetic flux 2-forms has attracted attention at least since Maxwell and W. Thompson. It turned out that the hydrodynamic analogue equations derived in \eqref{cH+bP-geom} provided the key to understanding why the introduction of noise into the Born-Infeld equations should appear as stochastic Lie transport. The hydrodynamic analogue for the Born-Infeld equations is summarised in the following theorem. 

\begin{theorem}[Hydrodynamic analogy for the stochastic Born-Infeld equations]\label{Stoch-Hydro-Analog}
The stochastic Born-Infeld equations in \eqref{BI-eqns-stoch-Lie} preserve both the conservative form for the evolution of $\cH$ and the Kelvin circulation theorem for $\bv = \bP/\cH$, as
\begin{align}
\begin{split}
&{\dd} \cH + {\rm div}\,\left(\cH \bsvt \right) = 0
\,,\\& 
\dd\oint_{c(\tilde{v})} (\bv\cdot d\bx ) 
+ \oint_{c(\tilde{v})} \big(\bsgamma\times {\rm curl}\bsgamma 
+  \bsbeta\times {\rm curl}\bsbeta \big) \cdot d\bx
= 0
\,,
\end{split}
\label{cH+bP-geom-thm}
\end{align}
for a closed loop $c(\tilde{v})$ moving with the stochastically augmented velocity $\bsvt :=\bv dt + \sum_i \bsxi_i \circ dW^i_t$. 
\end{theorem}
\begin{remark}\label{cons-erg-mom}
The first equation in \eqref{cH+bP-geom-thm} of Theorem \ref{Stoch-Hydro-Analog} implies that the Born-Infeld energy Hamiltonian $\bbH$ in \eqref{BI-Ham} for the deterministic augmented Born-Infeld equations \eqref{BIE-set} \emph{remains conserved} after introducing the stochastic Hamiltonian density in \eqref{h-cyl-noise} which pairs the noise with the Poynting vector. 

The second equation in \eqref{cH+bP-geom-thm} of Theorem \ref{Stoch-Hydro-Analog} shows that the hydrodynamic analogy of the ABI equations via the Kelvin circulation theorem for the deterministic case in \eqref{BI-KelThm} persists for the stochasticity introduced here. 
\end{remark}
\begin{proof}
After a short calculation, equations \eqref{BI-eqns-stoch-Lie} imply
\begin{align}
\begin{split}
&{\dd} \cH + \sum_i \cL_{\xi_i \circ dW^i_t}\cH 
+ {\rm div}\,\bP\,dt = 0
\,,\\& 
{\dd} \bP + \sum_i \cL_{\xi_i \circ dW^i_t}\bP
+ {\rm div}\,
\left( \frac{\bP\otimes\bP}{\cH} - \frac{\bD\otimes\bD}{\cH} - \frac{\bB\otimes\bB}{\cH}
\right)dt
= \nabla \left( \frac{1}{\cH}\right)dt
\,.
\end{split}
\label{cH+bP-stoch1}
\end{align}
Here, $\cL_{\xi \circ dW_t}\cH $ and $\cL_{\xi \circ dW_t}\bP$ denote, respectively,  the coefficients in the following Lie derivatives,
\begin{align}
\begin{split}
\cL_{\xi \circ dW_t}(\cH \, d^3x)\,,
&= 
\left({\rm div}(\cH \bsxi)\circ dW_t\right)\, d^3x
\\
\cL_{\xi \circ dW_t}(\bP\cdot d\bx\otimes d^3x)
&=
\left(\big(\partial_j(\xi^j P_k) + P_j \partial_k \xi^j \big)\circ dW_t\right)dx^k\otimes d^3x
\,.
\end{split}
\label{cH+bP-stoch2}
\end{align}
Since the Lie derivative of a scalar density is a divergence, the first equation in \eqref{cH+bP-stoch2} implies that conservation of the energy $\bbH (\bB,\bD,\bP) = \int \cH(\bB,\bD,\bP) \,d^3x$ defined in equation \eqref{BI-Ham-P} for the deterministic ABI persists in the stochastic case, provided the normal components $\mathbf{\hat{n}}\cdot\bsxi_i$ of the spatially dependent eigenvectors $\bsxi_i(\mathbf{x})$ do not contribute on the boundary of the domain of flow. 

After substituting $\bv = \bP/\cH$ and using \eqref{cH+bP-stoch2} to rearrange equations \eqref{cH+bP-stoch1}, we find
\begin{align}
\begin{split}
\Big(\dd + \cL_{vdt + \sum_i \xi_i \circ dW^i_t}\Big)(\cH\,d^3x) &= 0
\,,\\
\Big(\dd + \cL_{vdt + \sum_i \xi_i \circ dW^i_t}\Big)(\bv\cdot d\bx ) 
- \cL_\gamma (\bsgamma\cdot d\bx)
- \cL_\beta (\bsbeta\cdot d\bx)
&=
\frac12 d\left( \cH^{-2}  + |\bv|^2 - |\bsgamma|^2   - |\bsbeta|^2  \right)
\,.\end{split}
\label{cH+bP-stoch3}
\end{align} 
Upon introducing the vector field $\bsvt :=\bv dt + \sum_i \bsxi_i \circ dW^i_t$, the stochastic transport terms in these equations are expressed more compactly, as
\begin{align}
\begin{split}
(\dd + \cL_{\tilde{v}})(\cH\,d^3x) &= 0
\,,\\
(\dd + \cL_{\tilde{v}})(\bv\cdot d\bx ) 
- \cL_\gamma (\bsgamma\cdot d\bx)
- \cL_\beta (\bsbeta\cdot d\bx)
&=
\frac12 d\left( \cH^{-2}  + |\bv|^2 - |\bsgamma|^2   - |\bsbeta|^2  \right)
\,.\end{split}
\label{cH+bP-stoch4}
\end{align}
Returning from the Lie derivative notation to the coordinate form in vector notation as in equation \eqref{Lie-deriv-1form} and integrating the second equation in \eqref{cH+bP-stoch4} around a closed loop $c(\tilde{v})$ moving with the velocity $\bsvt(\bx,t)$ produces the equations \eqref{cH+bP-geom-thm} in the statement of the theorem. 

\end{proof}

\begin{remark}
Inserting the second equation in \eqref{cH+bP-stoch2} into the second equation in \eqref{cH+bP-stoch1} implies that conservation of total momentum $\int\bP \,d^3x$ obtained in equation \eqref{cH+bP-cons} for the deterministic ABI \emph{does not persist} in the stochastic case, unless $\partial_k \xi^j =0$, i.e., unless the amplitude of the noise is constant. 
\end{remark}

\begin{remark}[Fluid interpretation of stochastic augmented Born-Infeld equations \eqref{BI-eqns-stoch-Lie} \& \eqref{cH+bP-stoch4}] 

The stochastic ABI equations may be expressed in the same Lie-Poisson Hamiltonian form as for the deterministic case, which is introduced Appendix \ref{sec-var-form}, as
\begin{align}
\begin{split}
\dd
\begin{bmatrix}
P_i  \\ \bB \\ \bD
\end{bmatrix}
&=
\left\{
\begin{bmatrix}
P_i  \\ \bB \\ \bD
\end{bmatrix},
\widetilde{\bbH}
\right\}
\quad\hbox{which, upon substituting ${\rm div}\bB=0$ and ${\rm div}\bD=0$, becomes,}
\\&= - 
 \begin{bmatrix}
  ( P_j\partial_i + \partial_j P_i)\Box &   \bB\times{\rm curl}  \Box & \bD\times{\rm curl}\Box  
   \\
   -\,{\rm curl} (\,{\Box} \times \bB) & 0 & {\rm curl}{\Box}
   \\
   -\,{\rm curl} (\,{\Box} \times \bD)  & -\,{\rm curl}{\Box} & 0
   \end{bmatrix}
\begin{bmatrix}
{\delta \widetilde{\bbH}/\delta P_j} \\
{\de \widetilde{\bbH}}/{\de \bB} \\
{\de \widetilde{\bbH}}/{\de \bD} 
\end{bmatrix}.
\end{split}
\label{EM-LPB1}
\end{align}
Here, the stochastic Hamiltonian $\widetilde{\mathbb{H}}$ is given by
\begin{align}
\widetilde{\bbH}  (\bB,\bD,\bP) = \int \widetilde{\cH}(\bB,\bD,\bP) \,d^3x
\,,
\label{BI-Ham-tilde}
\end{align}
where $\widetilde{\cH}(\bB,\bD,\bP)$ in \eqref{BI-Ham-tilde} is defined by 
\begin{align}
\widetilde{\cH}(\bB,\bD,\bP) 
= \sqrt{1 + |\bD|^2 + |\bB|^2 + |\bP|^2}\,dt 
+ \sum_i \bP(\bx,t) \cdot \bsxi_i(\mathbf{x})\circ dW^i_t
\,.
\label{BI-Ham-P1}
\end{align}
The variational derivative ${\delta \widetilde{\bbH} }/{\delta \bP}$ of the stochastic Hamiltonian $\widetilde{\mathbb{H}}$ in \eqref{BI-Ham-tilde} recovers the stochastic velocity introduced in Theorem \ref{Stoch-Hydro-Analog},
\begin{align}
\frac{\delta \widetilde{\bbH} }{\delta \bP} = 
\bsvt :=(\bP/\cH) dt + \sum_i \bsxi_i(\bx) \circ dW^i_t
\,,
\label{BI-Ham-P2}
\end{align}
where again one takes variational derivatives in $(\bP,\bB,\bD)$ independently, as in \eqref{LP-form}. In particular, ${\delta \widetilde{\bbH} }/{\delta \bB}=\bB/\cH$ and ${\delta \widetilde{\bbH} }/{\delta \bD}=\bD/\cH$.
\end{remark}

{\bf The role of the momentum map.} So, why did the introduction of noise in both fluid dynamics and the Born-Infeld equations result in the same effect of introducing stochastic Lie transport into the motion equations for both particles and fields? The answer lies in the momentum map which relates the Hamiltonian structures of the two  theories. Namely, although their Poisson brackets are different, the two theories are related by the hydrodynamics analogy in \eqref{cH+bP-geom} comprising the momentum map given by the definition of the Poynting vector $\bP=\bD\times\bB$, regarded as a momentum 1-form density. 

Indeed, as we show in Appendix \ref{sec-var-form}, the Born-Infeld equations \eqref{BIE-set} for the $\bB$ and $\bD$ fields, augmented by the local conservation laws\eqref{cH+bP-cons} for the Poynting momentum $\bP$, together comprise Hamiltonian  dynamics on a Poisson manifold $\mathfrak{X}^*(\mathbb{R}^3)\times T^*\mathcal{C}^\infty(\mathbb{R}^3)$, whose Poisson structure is given by the sum of the canonical Poisson bracket for the electromagnetic fields, plus a Lie--Poisson bracket which is dual in the sense of $L^2$ pairing to the semidirect-product Lie algebra $\mathfrak{X}\,\circledS\,(\Lambda^1\otimes\Lambda^1)$, with dual coordinates $P\in\mathfrak{X}^*$, $B\in\Lambda^2$ and $D\in\Lambda^2$. This sum of a canonical Poisson bracket and a semidirect-product Lie-Poisson bracket derives from the definitions of the magnetic field flux $B:=dA$ and the cotangent lift momentum map $P=D\diamond A$ in \eqref{P-xi-pairing}. We refer to such augmented Poisson structures as KM brackets, after \cite{KrMa1987}.  The KM bracket for ABI derived in Appendix \ref{sec-var-form} reveals why the introduction of stochasticity by adding  to the deterministic Hamiltonian the stochastic term $\langle P,\xi(x) \rangle \circ dW_t$ in equation \eqref{LieBrkt-def} simply introduces a Lie derivative stochastic transport term in the resulting SPDE. In particular, introducing stochasticity in the transport of the Poynting vector momentum density in the ABI corresponds to introducing stochastic transport in both the displacement current and magnetic induction rate for the original Born-Enfeld field equations. 

{\bf Applications.} 
For both fluids and electromagnetic fields, the modified Hamiltonian which added the variational noise contribution was constructed by pairing eigenvector fields ostensibly describing the spatial correlations of the stochasticity in the data as cylindrical noise, with the momentum 1-form density variable for either the fluid or the fields. The infinitesimal Hamiltonian flow generated by the $L^2$ pairing of the momentum 1-form density with the sum of Stratonovich cylindrical noise terms with eigenvector fields derived from the spatial correlations of the data in both cases turned out to be an infinitesimal stochastic diffeomorphism. In fact, the infinitesimal map which resulted from the original Poisson structure was a sum of Lie derivatives with respect to the correlation eigenvector fields $\bsxi_i(\bx)$ for each component of the cylindrical noise. 

Thus, we conclude that uncertainty quantification for Hamiltonian systems can be based on stochastic Hamiltonian flows that are obtained from coupling the momentum map for the deterministic system with the sum over stochastic Stratonovich Brownian motions, formulated as cylindrical noise terms for each fixed spatial correlation eigenvector, as determined from the data being simulated. 

The application of the stochastic fluid dynamics treated here for quantifying uncertainty will depend crucially on determining the spatial correlations of the cylindrical noise represented by the eigenvectors $\bsxi_i(\bx)$ in Theorem \ref{Stoch-Hydro-Analog}. Extensive examples in fluid dynamics of how to obtain the spatial structure $\bsxi_i(\bx)$ of the cylindical noise and thereby quantify uncertainty in low resolution numerical simulations by comparing them to high resolution results for the same problem are given in \cite{CoCrHoShWe2018a,CoCrHoShWe2018b} for fluid dynamics simulations. These results should also be helpful examples for applying similar methods to the  Born-Infeld equations.


\subsection*{Acknowledgements}
I am grateful to Y. Brenier for illuminating discussions of the Born-Infeld model. I am grateful to D. O. Crisan, V. Putkaradze, T. S. Ratiu and C. Tronci, for encouraging and incisive remarks in the course of this work. I'd also like to thank the anonymous referees for their constructive suggestions. 
Finally, the author is also grateful to be partially supported by the European Research Council Advanced Grant 267382 FCCA and EPSRC Standard Grant EP/N023781/1. 

\appendix

\section{Variational formulations of ABI}\label{sec-var-form}

\subsection{Lagrangian formulation} 

Begin with Hamilton's principle, with action integral given by

\begin{align}
S(\bE,\bA; \bD) &:= \int_a^b  \mathcal{L}\, dt = \int_a^b \ell(\bE,\bB) dt\, + \int_a^b -\bD\cdot (\p_t\bA + \bE - \nabla\phi) + (\rho\phi - \bJ\cdot\bA) \,d^3x\,dt
\,,
\label{BI-action}
\end{align}
where $\bB={\rm curl}\bA$ is the magnetic field, $\bA(\bx,t)$ is the magnetic vector potential and $(\rho\phi - \bJ\cdot\bA)$ represents coupling of the electro magnetic fields to charged particle motion, represented by prescribed functions $\rho$ and $\bJ$, whose space and time dependence will satisfy the compatibility condition in \eqref{Max-gauge} for charge conservation, derived from gauge invariance of the Born-Infeld action $S$ in \eqref{BI-action}. Invoking Hamilton's principle and taking variations in \eqref{BI-action} yields 
\begin{align}
\begin{split}
0 = \delta S =&  \int_a^b \left\langle \frac{\de \ell}{\de \bE} - \bD\,,\,\de \bE\right\rangle
+ \left\langle {\rm curl}\frac{\de \ell}{\de \bB} +\p_t\bD - \bJ \,,\,\de \bA\right\rangle
+\langle - {\rm div} \bD + \rho\,,\, \de\phi  \rangle
\\&
-\langle \de \bD\,,\, \p_t\bA + \bE - \nabla\phi  \rangle \,dt
-\langle \bD\,,\, \de\bA   \rangle\Big|_a^b
\,,
\end{split}
\label{BI-HP1}
\end{align}
where angle brackets $\langle \, \cdot\,,\, \cdot\,  \rangle$ denote $L^2$ pairing. 
Upon recalling that $\bB={\rm curl}\bA$, stationarity $\de S=0$ of the action $S$ in \eqref{BI-action} implies the equations
\begin{align}
\begin{split}
\p_t \bB &= - \,{\rm curl}\,\bE 
\qquad
{\rm div} \bB = 0
\,,
\\
\p_t \bD &= -\,{\rm curl}\,\frac{\de \ell}{\de \bB} + \bJ
\qquad
{\rm div} \bD = \rho
\,.
\end{split}
\label{Maxwell-eqs}
\end{align}
These recover Maxwell's equations, upon defining $-\,{\rm curl}\,\frac{\de \ell}{\de \bB} =: {\rm curl}\,\bH$ and assuming a linear relation $\bH=\mu \bB$, between magnetic field $\bB$ and the magnetic induction, $\bH$. 

The fields $\bB={\rm curl}\bA$ and $\bE= - \p_t\bA + \nabla\phi $ in the Lagrangian in \eqref{BI-action} are invariant under gauge transformations given by
\begin{align}
\de \bA =  \nabla \psi \quad\hbox{and}\quad \de \phi = \p_t \psi
\quad\Longrightarrow\quad
\de \bE =0  \quad\hbox{and}\quad \de\bB=0
\,,
\label{Max-gauge}
\end{align}
for an arbitrary function $\psi(\bx,t)$.
Under these gauge transformations, the variation $\delta S$ of the action integral in \eqref{BI-HP1} will vanish after integration by parts, provided the matter fields $\rho$ and $\bJ$ satisfy the relation for  conservation of charge,
\begin{align}
\p_t\rho = {\rm div}\bJ 
\,,
 \label{Max-charge}
\end{align}
which is equivalent to preservation in time of the Gauss equation, ${\rm div} \bD = \rho$  in \eqref{Maxwell-eqs}, which is assumed to hold initially. 

\subsection{Noether momentum map} 

Now consider the Noether endpoint term  for the case that
\begin{align}
-\langle \bD\,,\, \de\bA   \rangle := -\int \bD\cdot\de\bA\,d^3x = \int \bD\cdot d\bS\wedge \pounds_\eta(\bA\cdot\,d\bx) 
\,,
 \label{Noether-endpt1}
\end{align}
where $\de\bA=-\,\pounds_\eta(\bA\cdot\,d\bx)$ denotes (minus) the Lie derivative of the 1-form $\bA\cdot\,d\bx$ by the smooth vector field $\eta=\bseta\cdot\nabla\in\mathfrak{X}(\mathbb{R}^3)$. Thus, 
\begin{align}
\begin{split}
-\int \bD\cdot\de\bA\,d^3x &= \int \bD\cdot \Big( - \bseta\times {\rm curl} \bA + \nabla (\bseta\cdot\bA)\Big)d^3x
\\
&= \int \bseta\cdot \Big( \bD \times {\rm curl} \bA -\bA {\rm div}\bD\Big)d^3x
=: \int \bseta\cdot \bP\,d^3x
\,,
\end{split}
 \label{Noether-endpt2}
\end{align}
where $\bP$ is the Poynting vector, defined as 
\begin{align}
\bP:= \bD \times \bB -\bA \,{\rm div}\bD\,.
 \label{Poynting-vec-def1}
\end{align}
Geometrically, the Poynting vector, is a 1-form density 
\begin{align}
P:= \bP \cdot d\bx \otimes d^3x = (\bD \times \bB -\bA \,{\rm div}\bD)\cdot d\bx \otimes d^3x
\,.
 \label{Poynting-vec-def2}
\end{align}
The expression \eqref{Poynting-vec-def1} for the Poynting vector reduces to the one in \eqref{P-xi-pairing} when ${\rm div}\bD=0$.

\subsection{Hamiltonian formulation} 

Legendre transforming the constrained Lagrangian $\mathcal{L}$ defined in  equation \eqref{BI-action} gives the following Hamiltonian,
\begin{align}
\begin{split}
{H}(\bD,\bA) &= \langle \bsPi\,,\,\partial_t \bA \rangle - \ell (\bE\,,\,\bB) 
+ \langle \bD\,,\,\p_t\bA + \bE - \nabla\phi \rangle
- \int (\rho\phi - \bJ\cdot\bA) \,d^3x
\\&= 
- \ell (\bE\,,\,\bB) 
+ \langle \bD\,,\, \bE - \nabla\phi \rangle
- \int (\rho\phi - \bJ\cdot\bA) \,d^3x
\,.
\end{split}
 \label{Legendre-Ham}
\end{align}
The canonical momentum is defined by $\bsPi=\de \mathcal{L}/\de( \p_t\bA) = -\bD$, so the two terms in  $\p_t\bA$ in the first line of \eqref{Legendre-Ham} have cancelled in the second line. Taking variations yields the following canonical equations,
\begin{align}
\begin{split}
\p_t \bA &= \frac{\de {H}}{\de \bsPi} = \frac{\de {H}}{\de (-\bD)} = -\,\bE + \nabla\phi\,,
\\
\p_t \bD &= - \,\p_t \bsPi = \frac{\de {H}}{\de \bA} = \bJ - {\rm curl} \frac{\de \ell}{\de \bB}
= \bJ + {\rm curl}\bH\,.
\end{split}
 \label{Dyn-EM-eqs}
\end{align}
These equations, along with the dynamical constraints 
\begin{align}
{\rm div}\bD=\rho
\quad\hbox{and}\quad
{\rm div}\bB=0
\,,
 \label{b+D-constrain}
\end{align}
will yield the equations of electromagnetism corresponding to any choice of the Lagrangian $\ell (\bE\,,\,\bB)$, or the Hamiltonian ${H}(\bD,\bA)$, related to each other by the Legendre transformation \eqref{Legendre-Ham}, so long as it is invertible. The evolution equations for the Poynting vector $\bP$ in \eqref{Poynting-vec-def1} and $\bB$ and $\bD$ in equations \eqref{Dyn-EM-eqs} may be written in 
Lie-Poisson Hamiltonian form as
\begin{align}
\begin{split}
\frac{\partial}{\partial t}
\begin{bmatrix}
P_i  \\ \bB \\ \bD
\end{bmatrix}
&=
\left\{
\begin{bmatrix}
P_i  \\ \bB \\ \bD
\end{bmatrix},
H
\right\}
\quad\hbox{which, upon substituting ${\rm div}\bB=0$, becomes, cf. \eqref{BI-hi-LPB},}
\\&= - 
 \begin{bmatrix}
  ( P_j\partial_i + \partial_j P_i)\Box &   \bB\times{\rm curl}  \Box & \bD\times{\rm curl}\Box  - ({\rm div}\bD) {\Box} 
   \\
   -\,{\rm curl} (\,{\Box} \times \bB) & 0 & {\rm curl}{\Box}
   \\
   -\,{\rm curl} (\,{\Box} \times \bD) + ({\rm div}\bD) {\Box} & -\,{\rm curl}{\Box} & 0
   \end{bmatrix}
\begin{bmatrix}
{\delta H/\delta P_j} \\
{\de H}/{\de \bB} \\
{\de H}/{\de \bD} 
\end{bmatrix},
\end{split}
\label{EM-LPB2}
\end{align}
or, in more geometrical form as, cf. equation \eqref{LP-form},
\begin{align}
\frac{\partial}{\partial t}
\begin{bmatrix}
P  \\ B \\ D
\end{bmatrix}
&=
\left\{
\begin{bmatrix}
P \\ B \\ D
\end{bmatrix},
H
\right\}
= - 
 \begin{bmatrix}
  {\rm ad}^*_\Box P &  \Box\diamond B & \Box\diamond D
   \\
   \pounds_{\Box} B & 0 & d{\Box}
   \\
   \pounds_{\Box} D & -d{\Box} & 0
   \end{bmatrix}
\begin{bmatrix}
{\delta H/\delta P} \\
{\de H}/{\de B} \\
{\de H}/{\de D} 
\end{bmatrix}.
\label{EM-LPB3}
\end{align}
Note that the variational derivatives of $H$ in $\bB$,  $\bD$ and $\bP$ appearing in \eqref{EM-LPB1}, \eqref{EM-LPB2} and \eqref{EM-LPB3} are to be taken independently. 

Expanding out the matrix differential operations and using the equivalence between coadjoint action and Lie derivative for vector fields acting 1-form densities yields,
\begin{align}
\begin{split}
(\p_t + \pounds_{\de H/\de P}) P &= - \frac{\de H}{\de B}\diamond B - \frac{\de H}{\de D}\diamond D 
\,,\\
(\p_t + \pounds_{\de H/\de P}) B &= -\,d\,\frac{\de H}{\de D}
\,,\\
(\p_t + \pounds_{\de H/\de P}) D &= d\,\frac{\de H}{\de B}
\,,\end{split}
\label{EM-geom}
\end{align}
where $\pounds_{\delta \mathbb{H}/\delta P}P$ is the Lie derivative of the 1-form density $P\in\mathfrak{X}^*$ by the vector field $v=\de \mathbb{H}/\delta P\in\mathfrak{X}$; the fluxes $B=\bB\cdot d\bS$ and $D=\bD\cdot d\bS$ are 2-forms in $\Lambda^2(\mathbb{R}^3)$; $\de H/\de B$ and  $\de H/\de D$ are 1-forms; and the diamond $(\diamond)$ operator is defined in the last line of \eqref{P-xi-pairing}.  

The augmented Hamiltonian structure in \eqref{EM-LPB1}, or equivalently \eqref{EM-LPB2} has made  $(P,A,D)\in\mathfrak{X}^*(\mathbb{R}^3)\times T^*\mathcal{C}^\infty(\mathbb{R}^3)$ into a Poisson manifold, whose Poisson structure is given by the sum of the canonical Poisson bracket plus a Lie--Poisson bracket dual in the sense of $L^2$ to the semidirect-product Lie algebra $\mathfrak{X}\,\circledS\,(\Lambda^1\otimes\Lambda^1)$ with dual coordinates $P\in\mathfrak{X}^*$, $B\in\Lambda^2$ and $D\in\Lambda^2$. This sum of a canonical Poisson bracket and a semidirect-product Lie-Poisson bracket derives from the definitions of the magnetic field flux $B:=dA$ and the cotangent lift momentum map $P=D\diamond A$. The augmented bracket in \eqref{EM-LPB2} reveals why the introduction of stochasticity by adding  to the deterministic Hamiltonian the stochastic term $\langle P,\xi(x) \rangle \circ dW_t$ in equation \eqref{LieBrkt-def} introduces a Lie derivative term in the resulting SPDE. 
We call such augmented Poisson structures KM brackets, after \cite{KrMa1987}.

\section{High field pressureless MHD limit of the ABI equations}\label{sec-MHD-form}

Brenier in \cite{Br2004} discusses a ``high field limit'' of the deterministic augmented Born-Infeld (ABI) equations, given in conservative form by
\begin{align}
\begin{split}
&\partial_t h + {\rm div}\,\bP = 0
\,,\\&
\partial_t \bP + {\rm div}
\left( \frac{\bP\otimes\bP}{h} - \frac{\bB\otimes\bB}{h}
\right)
= 0
\,,\\&
\partial_t \bB - {\rm curl}\left( \frac{ \bP\times \bB}{h} \right) = 0
\,,\\&
\hbox{with}\quad {\rm div}\bB = 0
\quad\hbox{and}\quad
\bP\cdot\bB = 0
\,.
\end{split}
\label{ABI-MHD-cons}
\end{align}
Equations \eqref{ABI-MHD-cons} have several similarities with magnetohydrodynamics (MHD), except that the energy density
$h=\sqrt{P^2+B^2}$, with total energy $\bbH=\int h \,d^3x$, differs from MHD. 
As in \eqref{cH+bP-geom} one may write these equations in geometric Lie transport form, 
upon introducing some additional notation. First, $\bv:=\de  \bbH/\de \bP=\bP/h$ and $\bsbeta:=\bB/h$ are components of vector fields,  and $\cL_v$ denotes the Lie derivative with respect to the vector field $v$, whose vector components are given in Euclidean coordinates by, e.g., $v=\bv\cdot\nabla=v^j\partial_j$. See equation \eqref{Lie-deriv-1form} for the vector components of the Lie derivative $\cL_v(\bv\cdot d\bx )$. Equations \eqref{ABI-MHD-cons} are then given equivalently by
\begin{align}
\begin{split}
(\partial_t + \cL_v)(h\,d^3x) &= 0
\,,\\
(\partial_t + \cL_v)(\bv\cdot d\bx ) &= \cL_\beta(\bsbeta\cdot d\bx ) - d|\bsbeta|^2
= - \bsbeta\times {\rm curl}\bsbeta \cdot d\bx 
\,,\\
(\partial_t + \cL_v)(\bB\cdot d\bS ) &= 0
\,,\\
\hbox{with}\quad {\rm div}\bB = 0
\quad\hbox{and}
&\quad
\bP\cdot\bB = 0
\,.\end{split}
\label{BI-hi-field-geom}
\end{align}
The second equation in \eqref{BI-hi-field-geom} looks like the motion equation for a pressureless version of classical MHD, in which $- \bsbeta\times {\rm curl}\bsbeta$ on its right hand side is to be regarded as the $\mathbf{J}\times\bB$ force.  

The high-field limit ABI MHD equations in conservative form \eqref{ABI-MHD-cons} and geometric form \eqref{BI-hi-field-geom} preserve  the constraints ${\rm div}\bB = 0$ and $\bP\cdot\bB=0$, provided they hold initially and $h$ is finite. Preservation of the constraint ${\rm div}\bB = 0$ is obvious. However, preservation of the constraint $\bP\cdot\bB=0$ requires a short calculation, as
\begin{align}
\begin{split}
&(\partial_t + \cL_v)\Big((\bv\cdot d\bx )\wedge (\bB\cdot d\bS )\Big)
=
(\partial_t + \cL_v) \big(\bv\cdot\bB \,d^3x\big)
\\&=
 - \bsbeta\times {\rm curl}\bsbeta \cdot d\bx \wedge (\bB\cdot d\bS )
 + (\bv\cdot d\bx )\wedge(\partial_t + \cL_v) (\bB\cdot d\bS )
\\&= 0
\\&=
\Big( \partial_t(\bv\cdot\bB) 
+ {\rm div}\big(\bv(\bv\cdot\bB)\big) \Big)\,d^3x 
\,.\end{split}
\label{BI-hi-constraint}
\end{align}
Consequently, we have a continuity equation for the quantity $\bv\cdot\bB$; namely
\begin{align}
 \partial_t(\bv\cdot\bB) 
+ {\rm div}\big(\bv(\bv\cdot\bB)\big)= 0
\,.
\label{BI-hi-constraint}
\end{align}
Thus, if the quantity $\bv\cdot\bB = \bP\cdot\bB/h$ vanishes initially, and $h$ does not vanish, then the constraint $\bP\cdot\bB=0$ is preserved by the ABI equations in \eqref{ABI-MHD-cons}, or equivalently, \eqref{BI-hi-field-geom}.

Since the Lie derivative commutes with the spatial differential, the differential of the second equation in \eqref{BI-hi-field-geom} immediately implies for the ABI \emph{vorticity} $\bom$, defined by 
\begin{align}
\bom\cdot d\bS = {\rm curl}\bv \cdot d\bS = d(\bv\cdot d\bx ) = d\big((\bP/h)\cdot d\bx \big) 
\,,\label{BI-def-omega}
\end{align}
that the evolution of ABI vorticity $\bom$ is given by 
\begin{align}
(\partial_t + \cL_v)(\bom\cdot d\bS ) = - {\rm curl}\big(\bsbeta\times {\rm curl}\bsbeta\big) \cdot d\bS \ne 0
\,.
\label{BI-hi-field-omega}
\end{align}
Consequently, the flux of ABI vorticity $\bom\cdot d\bS={\rm curl}(\bP/h) \cdot d\bS$ for hi-field MHD ABI is not frozen into the flow, as it is for Euler fluid flow.

The last equation in \eqref{BI-hi-field-geom} implies conservation of the linking number known as the \emph{magnetic helicity} $\Lambda_{mag}$ defined for $\bA={\rm curl}^{-1}\bB$ and ${\rm div}\bB=0$ by
\begin{align}
\Lambda_{mag}&=\int \lambda_{m}\,d^3x=\int \bA\cdot\bB\,d^3x =\int \bA\cdot d\bx\wedge \bB\cdot d\bS \,.
\label{BI-hi-mag-helicity}
\end{align}
This is the well-known topological winding number of ideal MHD. 
\begin{remark}
In terms of the energy Hamiltonian 
\begin{align}
\bbH = \int\!\! \sqrt{P^2 + B^2}\,d^3x =: \int  h ( \bP,\bB)\,d^3x 
\label{BI-hi-field-Ham}
\end{align}
we may write the deterministic MHD ABI equations \eqref{BI-hi-field-geom} in Lie-Poisson Hamiltonian form \cite{HoMaRa1998} as
\begin{align}
\frac{\partial}{\partial t}
\begin{bmatrix}
P_i  \\ \bB
\end{bmatrix}
=
\left\{
\begin{bmatrix}
P_i  \\ \bB
\end{bmatrix},
\bbH
\right\}
= - 
 \begin{bmatrix}
  ( P_j\partial_i + \partial_j P_i)\Box &   \bB\times{\rm curl}  \Box
   \\
   -\,{\rm curl} (\,{\Box} \times \bB)  & 0 
   \end{bmatrix}
\begin{bmatrix}
{\delta \bbH/\delta P_j} = P^j/h = v^j\\
{\de \bbH}/{\de \bB} = \bB/ h = \bsbeta
\end{bmatrix}
\,,
\label{BI-hi-LPB}
\end{align}
where the boxes $(\Box)$ indicate how the differential operators in the matrix elements act on the variational derivatives of the Hamiltonian in $P_j$ and $\bB$. (Note that these variations are taken independently.)
After a direct calculation, the Lie-Poisson bracket \eqref{BI-hi-LPB} recovers 
\begin{align}
\partial_th = \{h\,,\,\bbH\} = -\,{\rm div}\left(  \frac{P^2+B^2}{h} \bv\right) + {\rm div}(\bB(\bv\cdot\bB))
=  -\, {\rm div} \Big(\big(\sqrt{P^2 + B^2}\big)\bv\Big) =  -\, {\rm div}(h\bv)
\,,
\label{BI-hi-LPBh}
\end{align}
from orthogonality $\bv\cdot\bB = 0$ and the definition of $h$ in \eqref{BI-hi-field-Ham}. 
This calculation shows that orthogonality $\bP\cdot\bB = 0$ and the continuity equation for $h$, together with   
the skew adjointness of the Lie-Poisson bracket in \eqref{BI-hi-LPB}, ensure that $d\bbH/dt=\{\bbH,\bbH\}=0$ for appropriate boundary conditions and that the energy Hamiltonian in \eqref{BI-hi-field-Ham} is conserved for any choice of boundary conditions on $\bv$ and $\bB$.
\end{remark}
\bigskip

\paragraph{\bf Covariance upon introducing stochasticity.}
Stochastic Lie transport that is introduced by coupling the noise to the momentum map simply modifies the transport operator in equations \eqref{BI-hi-field-geom}, as
\begin{align}
(\partial_t + \cL_{v}) \to (\dd + \cL_{\tilde{v}})
\quad\hbox{with } \tilde{v} = \bsvt\cdot\nabla\hbox{ and}\quad
\bsvt :=(\bP/h) dt + \sum_i \bsxi_i (\bx)\circ dW^i_t
\,,
\label{BI-hi-field-stoch}
\end{align}
in which 
\begin{align}
\bsvt = \frac{\delta \widetilde{\bbH}}{\delta \bP}  
\quad\hbox{where}\quad
\widetilde{\bbH} := \int\!\! \sqrt{P^2 + B^2}\,d^3x\,dt + \int \bP\cdot \sum_i \bsxi_i (\bx)\,d^3x \,\circ dW^i_t
\,.
\label{BI-hi-field-vel}
\end{align}
Consequently, in the previous equations \eqref{BI-hi-field-geom} and \eqref{BI-hi-constraint} the terms with Lie derivative $ \cL_v$ are simply replaced with $ \cL_{\tilde{v}}$ when Lie transport stochasticity in equation \eqref{BI-hi-field-stoch} is introduced, while the other terms are left unchanged. Thus, in the high field limit of the augmented Born-Infeld equations, the introduction of Lie transport stochasticity preserves their geometric form. In addition, the Lie-Poisson bracket in equation \eqref{BI-hi-LPB} also persists under the introduction of noise this way, although the Hamiltonian becomes stochastic. 



\begin{thebibliography}{99}                                                                                               %

\bibitem{Arek1989}
M. Arik, F. Neyzi, Y. Nutku, P. J. Olver and J. M. Verosky [1989], Multi-Hamiltonian structure of the Born--Infeld equation. \textit{Journal of Mathematical Physics}, 30 (6): 1338-1344.

\bibitem{Bial1984}
I. Bialynicki-Birula [1984], Nonlinear Electrodynamics: Variations on a Theme by Born and Infeld
(Warsaw, CFT). In Jancewicz, B. ( Ed.), Lukierski, J. ( Ed.): \textit{Quantum Theory Of Particles and Fields}, 31-48.

\bibitem{Bial1992}
I. Bialynicki-Birula [1992], Field theory of photon dust. \textit{Acta Phys. Polon.}, 23: 553-559.

\bibitem{Bi1981}
 J. M. Bismut [1981], \textit{M\'ecanique al\'eatoire}, Berlin: Springer.

\bibitem{BoIn1934}
M. Born and L. Infeld [1934], Foundations of the new field theory,
\textit{Proceedings of the Royal Society A} 144: 425-451.

\bibitem{Br2002}
Y. Brenier [2002], A note on deformations of 2D fluid motions using 3D Born-Infeld equations, \textit{Nonlinear Differential Equation Models}, 2012 - books.google.com

\bibitem{Br2004}
Y. Brenier [2004], Hydrodynamic structure of the augmented Born-Infeld equations,
\textit{Archive for Rational Mechanics and Analysis}, 172 (1):  65-91.

\bibitem{BrYo2005}
Y. Brenier and W. A.Yong [2005], Derivation of particle, string, and membrane motions from the Born-Infeld electromagnetism. 
\textit{Journal of Mathematical Physics}, 46 (6): 062305.

\bibitem{CoCrHoShWe2018a}
C. J. Cotter, D. O. Crisan, D. D. Holm, I. Shevschenko and W. Pan [2018a], 
Modelling uncertainty using circulation-preserving stochastic transport noise in a 2-layer quasi-geostrophic model.
arXiv:1802.05711.

\bibitem{CoCrHoShWe2018b}
C. J. Cotter, D. O. Crisan, D. D. Holm, I. Shevschenko and W. Pan [2018b], 
Numerically Modelling Stochastic Lie Transport in Fluid Dynamics.
arXiv:1801.09729.

\bibitem{CrFlHo2017}
D. O. Crisan, F. Flandoli, D. D. Holm [2017],
Solution properties of a 3D stochastic Euler fluid equation.
\url{https://arxiv.org/abs/1704.06989}

\bibitem{Di1960}
P. A. M. Dirac [1960], A reformulation of the Born-Infeld electrodynamics, 
\textit{Proceedings of the Royal Society A}  257 (1288) 32-43.

\bibitem{Gi2001}
G. W. Gibbons [2001], Aspects of Born-Infeld theory and String/M-theory, 
\url{https://arxiv.org/abs/hep-th/0106059}.

\bibitem{HaJoSt2007}
A. Hannachi, I. T. Jolliffe, D. B. Stephenson [2007],
Empirical orthogonal functions and related techniques in atmospheric science: A review. \textit{International Journal of Climatology} 27 (9): 1119-1152.

\bibitem{Holm1987}
D. D. Holm [1987], 
Hamiltonian dynamics and stability analysis 
of neutral electromagnetic fluids with induction, 
{\it Physica D} 25: 261-287.

\bibitem{Ho2011} 
D. D. Holm [2011], \textit{Geometric Mechanics, Part 2}, World Scientific.

\bibitem{Holm2015}
D. D. Holm [2015], Variational principles for stochastic fluid
dynamics, \textit{Proceedings of the Royal Society A} 471: 20140963.

\bibitem{HoMaRa1998} 
D. D. Holm, J. E. Marsden and T. S. Ratiu [1998],
The Euler--Poincar\'e equations and semidirect products
with applications to continuum theories,
{\it Adv. in Math.}, {\bf 137}, 1-81,
arXiv e-print available at \url{http://xxx.lanl.gov/abs/chao-dyn/9801015}.

\bibitem{HoJiZh2014}
J. Hong, L. Ji and L. Zhang [2014],
A stochastic multi-symplectic scheme for stochastic Maxwell equations with additive noise, \textit{Journal of Computational Physics} 268: 255-268.

\bibitem{HoStYa2010}
T. Horsin, I. G. Stratis and A. N. Yannacopoulos [2010],
On the approximate controllability of the stochastic Maxwell equations,
\textit{IMA Journal of Mathematical Control and Information} 27: 103-118.

\bibitem{Kiessling2011}
M. K. H. Kiessling [2011] Some uniqueness results for stationary solutions to the Maxwell-Born-Infeld field equations and their physical consequences. Physics Letters A, 375(45), pp.3925-3930.

\bibitem{KrMa1987}
P. S. Krishnaprasad and J. E. Marsden [1987] Hamiltonian structures and stability for rigid bodies with flexible attachments. 
Archive for Rational Mechanics and Analysis, 98(1), pp.71-93.

\bibitem{MaRa1994} 
J. E. Marsden and T. S. Ratiu [1994], \textit{Introduction to Mechanics and Symmetry}, Springer-Verlag.

\bibitem{Pa2007} 
E. Pardoux [2007],
\textit{Stochastic Partial Differential Equations},
Lectures given in Fudan University, Shanghai. Published by Marseille, France.

\bibitem{Pauli1933}
W. Pauli [1933], \textit{General Principles of Quantum Mechanics}. English translation reprinted in 1981 by Springer-Verlag.

\bibitem{Pol1998}
J. Polchinski [1998], \textit{String theory. Vol. I}, Cambridge University Press, Cambridge.

\bibitem{Sc1988} 
K.-U. Schauml\"offel [1988], White noise in space and time and the cylindrical Wiener process, \textit{Stochastic Analysis and Applications}, 6:1, 81--89.

\bibitem{Si1991}
D. M. Siegel [1991], \textit{Innovation in Maxwell's Electromagnetic
Theory} (Cambridge: Cambridge University Press).
%
%
%

%



\end{thebibliography}
\end{document}